\documentclass[prl, aps,preprintnumbers,showpacs, superscriptaddress]{revtex4}
\usepackage{latexsym}
\usepackage{amsmath}
\usepackage{amssymb}
\usepackage{amsfonts}
\usepackage{ifthen}
\usepackage{revsymb}
\usepackage{yfonts}
\usepackage{graphicx}

\usepackage{graphicx}
\usepackage{amsmath}
\usepackage{amssymb}
\usepackage{amsthm}

\newcommand{\be}{\begin{equation}}
\newcommand{\ee}{\end{equation}}
\newcommand{\ba}{\begin{array}}
\newcommand{\ea}{\end{array}}
\newcommand{\bea}{\begin{eqnarray}}
\newcommand{\eea}{\end{eqnarray}}

\newcommand{\la}{\langle}
\newcommand{\ra}{\rangle}

\newcommand{\nn}{\nonumber}

\newtheorem{lemma}{Lemma}

\newtheorem{theorem}{Theorem}

\begin{document}

\title{Simulation of Many-Body Hamiltonians using Perturbation Theory with Bounded-Strength Interactions}

\author{Sergey \surname{Bravyi}
\footnote{sbravyi@us.ibm.com}} \affiliation{IBM Watson Research
Center, P.O. Box 218, Yorktown Heights, NY, USA}
\author{David P. \surname{DiVincenzo}
\footnote{divince@watson.ibm.com}} \affiliation{IBM Watson Research
Center, P.O. Box 218, Yorktown Heights, NY, USA}
\author{Daniel \surname{Loss}
\footnote{daniel.loss@unibas.ch}} \affiliation{Department of
Physics, University of Basel, Klingelbergstrasse 82, CH-4056 Basel,
Switzerland}
\author{Barbara M. \surname{Terhal}
\footnote{bterhal@gmail.com}} \affiliation{IBM Watson Research
Center, P.O. Box 218, Yorktown Heights, NY, USA}

\date{\today}

\begin{abstract}
We show how to map a given $n$-qubit target Hamiltonian with
bounded-strength $k$-body interactions onto a simulator Hamiltonian
with two-body interactions, such that the ground-state energy of the
target and the simulator Hamiltonians are the same up to an {\em
extensive} error $O(\epsilon n)$ for arbitrary small $\epsilon$. The
strength of interactions in the simulator Hamiltonian depends on
$\epsilon$ and $k$ but does not depend on $n$. We accomplish this
reduction using a new way of deriving an effective low-energy
Hamiltonian which relies on the Schrieffer-Wolff transformation of
many-body physics.
\end{abstract}


\maketitle

\section{Introduction}

In quantum field theory and quantum many-body physics perturbation
theory is omnipresent. Perturbation theory is perhaps most
important as a method for constructing an efficient description of
the physics at low energies as an effective theory of the physics
at higher energies. For systems which do not have a
perturbative treatment, it can be much harder or impossible to obtain
such a concise and efficient description. In an apparently
independent development, perturbation theory has recently appeared
as a tool in the area of quantum complexity theory
\cite{kkr:hamsiam}. Its use in this area resembles that in
theoretical physics: it relates the low-energy properties of
different Hamiltonian models. But in computer science this
association is made in the reverse direction as compared with
theoretical physics.  In many-body physics, the ``full"
Hamiltonian (including high energy degrees of freedom) is given,
and an effective Hamiltonian on the low-energy subspace is
calculated using perturbation theory. In computational complexity
theory, the effective Hamiltonian is given -- it is a ``target
Hamiltonian" $H_{\rm target}$ chosen for some computational
hardness property (for example, in adiabatic quantum
computation~\cite{ADKLLR} the ground state of $H_{\rm target}$
encodes the execution of a quantum algorithm). Then a
``high-energy" simulator Hamiltonian $H$ is to be chosen such that
$H_{\rm target}$ can be obtained from it by perturbation theory.
The main objective is to make the simulator Hamiltonian $H$ as
simple and realistic as possible while retaining the computational
hardness of $H_{\rm target}$.
For example, the perturbation theory has been used
in~\cite{OT:qma} to prove universality of quantum adiabatic
computation with local two-body Hamiltonians on a 2D square
lattice. For more recent developments
see~\cite{BDOT:stoq,JF:gadgets,ver,Love}.

In all applications mentioned above, one deals with a many-body
Hamiltonian $H=H_0+V$ where $H_0$ is a simple operator with a known
spectrum and $V$ is a perturbation. The spectrum of $H_0$ consists
of two disjoint parts spanning the low-energy subspace $P$ and the
high-energy subspace $Q$. The two parts of the spectrum are
separated from each other by an energy gap $\Delta$. A perturbation
expansion provides a systematic way of constructing an effective
Hamiltonian  $H_{\rm eff}$ acting on the low-energy subspace $P$
such that the ground state energy of $H_{\rm eff}$ approximates the
one of $H$ with an arbitrarily small error. An important shortcoming
of the perturbative approach is the limited range of parameters over
which it can be rigorously justified, namely, $\|V\|<\Delta/2$, see
for instance~\cite{kkr:hamsiam,book:reed&simon}. For $||V|| \gg
\Delta$, we expect states from the high-energy subspace to mix
strongly with states from the low-energy subspace. Thus, it should
not be expected that a general perturbative expansion for $H_{\rm
eff}$ is convergent when $||V|| \gg \Delta$. (One notable exception
to this generic behavior is the case when the low-energy subspace is
one-dimensional, see~\cite{BDL07}.)

If we require for convergence that $||V|| \ll \Delta$, we force
$\Delta$ to scale with system size $n$, since $||V||$ typically
scales with $n$; this is what has been done in
\cite{kkr:hamsiam,OT:qma,BDOT:stoq,JF:gadgets,ver,Love}. Hence
these calculations would indicate that perturbation theory can
only be rigorously applied to an unphysical Hamiltonian, one for
which the gap $\Delta$ grows with the system size $n$.

If the applicability of perturbation theory were thus limited, its
common use in physics would be unwarranted. The lack of convergence
of the perturbative series for quantum electrodynamics was argued on
physical grounds by Dyson \cite{dyson:pert}. There is a widespread
belief that general perturbative series for quantum field theory and
many-body physics do not converge but are to be viewed as {\em
asymptotic} series, meaning that their lowest-order terms are a good
approximation to the quantity of interest (while inclusion of
higher-order terms may actually give a worse result, see e.g.
\cite{book:reed&simon}.)

In this Letter we make these beliefs more rigorous by developing a
formalism that justifies application of perturbation theory in the
regime  $||V|| \gg \Delta$. This formalism is applicable to
many-body Hamiltonians that possess certain locality properties. A
Hamiltonian $H$ describing a system of $n$ qubits will be called
{\it $k$-local} if and only if it is represented as a sum of local
interactions $H=\sum_i H_i$ such that each operator $H_i$ acts on
some subset of $k$ or less qubits. In addition, the Hamiltonians in
this paper will have the important property that each qubit occurs
in at most a constant, say $m=O(1)$, number of terms $H_i$. The
constraint says that for growing system size, the number of
interactions in which any one qubit participates does not grow, but
stays constant. Such a condition is generically fulfilled for
Hamiltonians with shoft-range interactions studied in physics. An
example is the standard (2-local) Heisenberg Hamiltonian on a
lattice $H=-\sum_{i,j} J(r_{ij})(X_i X_j+Y_i Y_j+Z_i Z_j)$; the fact
that the coupling $J(r_{ij})$ is bounded-range, $J(r_{ij}
> r) \approx 0$, ensures that each spin is acted upon by a
constant number of terms. Note that our results are not restricted
to Hamiltonians on lattices, but also hold for general networks of
interactions such as expander graphs \cite{Vazirani}.

The {\it interaction strength} of a $k$-local Hamiltonian $H=\sum_i H_i$ is defined
as the largest norm of the local interactions, $J=\max_i \|H_i\|$.
Throughout this paper we assume that $k$ is a constant independent
of $n$.

Our main result is a rigorous bound on the error that is made when a
perturbative expansion for the effective Hamiltonian is truncated at
a low order (the second or the third) in the regime where the {\em
interaction strength} of $V$ is small compared to $\Delta$ but
$||V|| \gg \Delta$. 
This bound allows us to prove the following result.
\begin{theorem}
\label{thm:theorem3} Let $H_{\rm target}$ be a $k$-local Hamiltonian
acting on $n$ qubits with interaction strength $J$. For any fixed
precision $\epsilon$ one can construct a $2$-local simulator
Hamiltonian $H$ acting on $O(n)$ qubits with interaction strength
$O(J)$ such that the ground state energy of $H$ approximates the
ground state energy of $H_{\rm target}$ with an absolute error at
most $\epsilon Jn$.
\end{theorem}
Let us comment on a significance of Theorem \ref{thm:theorem3}.
First of all, it eliminates the essential shortcoming of the
earlier
constructions~\cite{kkr:hamsiam,OT:qma,BDOT:stoq,JF:gadgets,ver,Love}
mentioned above, namely, the unphysical scaling of the interaction
strength in the simulator Hamiltonian $H$.  In our construction
the interaction strength of $H$ is bounded by a constant
(depending only on $k$, $\epsilon$, and $J$) which makes it more
physical and, in principle, implementable in a lab. The price we
pay for this improvement is that  we are able to reproduce the
ground state energy only  up to an {\em extensive} error $\epsilon
J n$.  For all realistic physical Hamiltonians the ground-state
energy itself is generically proportional to $nJ$, and thus the
{\it relative error} can be made arbitrarily small. Secondly,
estimating the ground-state energy of a $k$-local Hamiltonian up
to some sufficiently small constant relative error is known to be
an NP-hard problem even for classical (diagonal in the
$|0\ra,|1\ra$ basis) Hamiltonians. This hardness of approximation
follows from the PCP theorem (PCP=probabilistically checkable
proofs), see Chapter~29 of~\cite{Vazirani} for a review. Proving
an analogous  hardness of approximation result for quantum
Hamiltonians is a widely anticipated development, a so called
``quantum PCP theorem", see Section~5 in~\cite{AGIK}. If true, it
would imply that reproducing the ground-state energy up to some
sufficiently small relative error (which is what our simulation
achieves) is enough to capture the computational hardness of the
target Hamiltonian. We also believe that  the simulation results
reported in this Letter will be relevant in actually proving such
a quantum PCP theorem. Finally, the techniques we develop may be
relevant in the context of adiabatic quantum computation where the
simulator Hamiltonian must be capable of reproducing ground-state
expectation values of local observables, see~\cite{OT:qma,JFS:ft}.


The proof of Theorem~1  relies on a new way of deriving an
effective Hamiltonian, which we call the {\em Schrieffer-Wolff
transformation} after its use in many-body physics~\cite{SW}.
Specifically, given an unperturbed Hamiltonian $H_0$, projectors
$P$, $Q$ onto the low-energy and the high-energy subspaces of
$H_0$, and a perturbation $V$, the Schrieffer-Wolff transformation
is a
unitary operator $e^S$ with anti-hermitian $S$ ($S=-S^{\dagger}$)
such that (i) $S$ is a block-off-diagonal operator, namely,
$PSP=0$ and $QSQ=0$; (ii) the transformed Hamiltonian $e^S\,
(H_0+V)\, e^{-S}$ is a block-diagonal operator, namely, $P\, e^S
(H_0+V)e^{-S}\, Q=0$. Given such an operator $S$, one defines the
effective Hamiltonian on the low-energy subspace as $H_{\rm
eff}=P\,   e^S (H_0+V)e^{-S}\, P$. A detailed exposition of the
Schrieffer-Wolff formalism will appear in~\cite{BDLT:pert_gen}.

Let us briefly sketch the proof of the theorem. The  simulator
Hamiltonian $H$ is constructed using {\em perturbation gadgets}
introduced in \cite{kkr:hamsiam,OT:qma}. ``Gadget" is a technical
term used broadly in theoretical computer science; in the present
application, a gadget is simply an extra {\it mediator qubit}, and
a Hamiltonian coupling the mediator qubit with some small subset
of system qubits. For every mediator qubit $u$ we define a
projector onto its low-energy subspace $P^u=|0\ra\la 0|_u$ and its
high-energy subspace $Q^u=|1\ra\la 1|_u$. The purpose of a gadget
is to simulate some particular $k$-body interaction $H_{\rm
target}^u$ in the decomposition $H_{\rm target} \equiv \sum_u
H_{\rm target}^u+H_{\rm else}$. Here $H_{\rm else}$ are additional
terms in the target Hamiltonian that we do not wish to treat using
perturbation theory (since they are already 2-local, for example).

Such simulation is achieved by applying perturbation theory to each
gadget individually.  A gadget's simulator Hamiltonian is

$H^u=H_0^u+V^u$, where $H_0^u=\Delta\, Q^u$ penalizes the mediator
qubit for being in the state $|1\ra$, and $V^u$ is a perturbation.
With the proper choice of $V^u$ the effective Hamiltonian on the
low-energy subspace, in which the mediator qubit $u$ is in the state
$|0\ra$, approximates $H_{\rm target}^u$ with an error $\epsilon$.
Furthermore, this effective Hamiltonian can be obtained from $H^u$
via an approximate Schrieffer-Wolff transformation, that is, $H_{\rm
target}^u=P^u \, e^{S^u}\, (H_0^u+V^u)\, e^{-S^u}\, P^u + O(\epsilon
J)$ for some anti-hermitian $S^u$ satisfying $P^u S^u P^u=Q^u S^u
Q^u=0$.

Combining the local gadgets together  we get as a candidate for the
simulator Hamiltonian $H=H_0+V+H_{\rm else}$ with $H_0=\sum_u H_0^u$
and $V=\sum_u V^u$. Let $\lambda(H)$ and $\lambda(H_{\rm target})$
be the ground-state energy of $H$ and $H_{\rm target}$ respectively.
We prove that $\lambda(H)$ approximates $\lambda(H_{\rm target})$
with a small extensive error by constructing a global unitary
transformation $e^S$ mapping $H$ to $H_{\rm target}$ (with a small
extensive error), that is, $H_{\rm target}=P\, e^S H e^{-S}\, P +
O(\epsilon n J)$ where $P=\bigotimes_u P^u$ projects onto the
subspace in which every mediator qubit is in the state $|0\ra$.
Given such a transformation one immediately gets an upper bound
$\lambda(H)=\lambda(e^S H e^{-S})\le \lambda(P\, e^S H e^{-S}\, P) =
\lambda(H_{\rm target}) + O(\epsilon n J)$. Here we have taken into
account that restricting a Hamiltonian on a subspace can only
increase its ground-state energy~\footnote{By abuse of notation,
$P\, e^S H e^{-S}\, P$ denotes the restriction of the operator $e^S
H e^{-S}$ onto the subspace $P$.}. Making a natural choice $S=\sum_u
S^u$ we prove that $P\, e^S H e^{-S}\, P$ contains the desired term
$H_{\rm target}$ {\it and} some cross-gadget terms where $S^u$ acts
on $H_0^v+V^v$ with $u\ne v$.  Using the 'independence' properties
of the gadgets and the block-off-diagonality of $S$ we are able to
show that the contribution of these cross-gadget terms are small
enough to be absorbed into the error term $O(n\epsilon)$,  see
Section~II.

In order to prove a matching lower bound, consider the transformed
Hamiltonian $\tilde{H}^u=e^{S_u} H^u e^{-S_u}$, where
$H^u=H_0^u+V^u$. We prove an operator inequality $\tilde{H}^u \ge
I^u\otimes H_{\rm target}^u + O(\epsilon J)$.
 Here $I^u$ is the identity operator acting on the mediator qubit
$u$ and  $O(\epsilon J)$ stands for some operator with norm
$O(\epsilon J)$. Intuitively one should expect this inequality to
be true since the $P$-block of $\tilde{H}^u$ approximates $H_{\rm
target}^u$ with an error $O(\epsilon J)$, the $Q$-block of
$\tilde{H}^u$ contains a large energy penalty $\Delta$, and the
off-diagonal blocks $P \tilde{H}^u Q$ are small by the definition
of the Schrieffer-Wolff transformation. Using the unitarity of
$e^{S^u}$ and smallness of  $S^u$ we transform the above
inequality into $H^u \ge I^u\otimes H_{\rm target}^u + O(\epsilon
J)$ which implies $H\ge I\otimes H_{\rm target} + O(n\epsilon J)$
and thus gives $\lambda(H)\ge \lambda(H_{\rm target}) +
O(n\epsilon J)$. These arguments are filled in at the beginning of
Section~II.

Our results will be stated for two different mappings $H_{\rm
target} \rightarrow H$ corresponding to two different gadgets, the
one reducing the locality parameter $k$  by a factor of $2$, and
the other reducing $k=3$ to $k=2$. By composing these mappings we
arrive at Theorem \ref{thm:theorem3}.


Before proceeding with the details, let us make a few remarks about
generalizations of this technique. We expect that similar results
can be obtained for other perturbation gadgets in the literature
(see e.g. \cite{kkr:hamsiam,JF:gadgets,kitaev:anyon_pert}), since
these gadgets are all {\em designed} to work independently, i.e. a
term in the target Hamiltonian is replaced by some local perturbed
Hamiltonian $H$ and cross-gadgets terms should have small
contributions in the perturbation expansion. For Hamiltonians $H$ of
direct interest in many-body physics, we expect that it will also be
possible to identify a gadget sub-structure $H=\sum_u H^u$ such that
$H^u$ gives rise to $H_{\rm target}^u$ {\em and} cross-gadget
contributions are small. For generic perturbed Hamiltonians,
cross-gadget terms may not be small, which implies that the effect
of the perturbations must be analyzed globally. We will consider
such an analysis in a future paper \cite{BDLT:pert_gen}.

Now we state two Lemmas used in the proof. The two Lemmas together
can be regarded as an infinitesimal version of the Lieb-Robinson
bound that governs time evolution of a local observable under a
local Hamiltonian, see e.g.~\cite{HK:liebrob,BHV:liebrob}.
\begin{lemma}
\label{prop:LR1} Let $S$ be an anti-hermitian operator. Define a
superoperator $L$ such that $L(X)=[S,X]$ and $L^0[X]=X$. For any
operator $H$ define $r_0(H)=\| e^S H e^{-S} \| =\| H\|$,
$r_1(H)=\| e^S H e^{-S} - H\|$, and \be r_k(H)=\|\, e^S H e^{-S} -
\sum_{p=0}^{k-1} \frac1{p!}\, L^p(H) \, \|, \quad k\ge 2, \ee
where $\|\cdot \|$ is the operator norm. Then for all $k\ge 0$ one
has \be r_k(H)\le \frac1{k!}\, \|\, L^k(H)\, \|. \ee
\end{lemma}


Using this Lemma one can show that
\begin{lemma}
\label{prop:LR2} Let $S$ and $H$ be any $O(1)$-local operators
acting on $n$ qubits with an interaction strengths $J_S$ and $J_H$
respectively. Let each qubit be acted upon non-trivially by $O(1)$
terms of $H$ and $S$. Then for any $k=O(1)$ one has \be
\label{LR2} \|\, L^k(H) \, \| =O(n \cdot J_S^k J_H). \ee
\end{lemma}

The proofs are rather elementary and can be found in the Appendix.

\section{I. The Gadgets}
We shall use two types of gadgets proposed in~\cite{OT:qma},
namely, the {\it subdivision} gadget and the {\it 3-to-2-local}
gadget. The former will be used to break $k$-local interactions
down to $3$-local interactions while the latter breaks $3$-local
interactions up into $2$-local interactions.  A notational
comment: in this section, all Hamiltonian operators, and
Schrieffer-Wolff operators $S$, refer to a single gadget $u$, and
in the remainder of the paper these operators appear with the
label $u$; in this section only, for economy, we omit this label.

\noindent {\bf Subdivision Gadget}. Let the target Hamiltonian be a single $k$-body interaction
$H_{\rm target}=J\, A B$ where $A$, $B$ act on non-overlapping
subsets of  $\lceil k/2\rceil $ or less qubits and $\|A\|,\|B\|\le
1$. Introduce one mediator qubit $u$, choose a parameter $\Delta
\gg J$,  and define the simulator Hamiltonian $H=H_0+V$ with \bea
 \label{H}
  H_0=\Delta |1\ra\la 1|_u, &
V=\sqrt{\Delta J/2}\, X_u \otimes (-A+B) +V_{\rm extra}. \eea
 Here $V_{\rm extra}=(J/2)(A^2+B^2)$ acts trivially on the mediator qubit.
 Note that  $H$ contains only
 $(\lceil k/2 \rceil+1)$-body interactions.
The  purpose of the term $V$  is to induce transitions $|0\ra_u\to |1\ra_u\to |0\ra_u$
in the second order of perturbation theory, see Fig.~1,
such that the corresponding effective Hamiltonian is proportional to $(-A+B)^2$ containing the
desired term $AB$ and  unwanted terms $A^2$, $B^2$ which we cancel by $V_{\rm extra}$.
Next we define
 \be
  \label{eq:s_sub}
  S=-i\sqrt{\frac{J}{2\Delta}}\,  {Y_u} \otimes (-A+B).
 \ee
One can check that $S$ is the Schrieffer-Wolff transformation
truncated at the first order,  i.e. $S=S_1$ where $S_1$ is defined
in Eq.~(\ref{S1S2S3}) in the Appendix. We can calculate $P e^S H
e^{-S}P$ using the expansion in Lemma \ref{prop:LR1}. A
straightforward calculation shows that
\bea
 \label{transformedH2}
e^{S}H e^{-S}=\left( H + [S,H] + \frac12
[S,[S,H]]\right)+O(J^{3/2}\Delta^{-1/2}) =\left[ \ba{cc} H_{\rm target} & 0 \\
0 & \Delta\, I + O(J) \\ \ea \right] +
O(J^{\frac32}\Delta^{-\frac12}), \eea where we used that
$[S,V_{\rm extra}]=0$. The upper and lower blocks correspond to the
subspaces $P$ and $Q=I-P$ respectively.  Thus, $P e^S H e^{-S}P$ is
close to $H_{\rm target}=J AB$, as desired; the error can be made
$O(\epsilon J)$ by choosing $\Delta=J\epsilon^{-2}$.

\noindent {\bf 3-to-2-local Gadget}. Let the target Hamiltonian be
a single 3-body interaction, $H_{\rm target}=J\, A BC$, where
$A,B,C$ are one-qubit operators acting on different qubits and
$\|A\|,\|B\|,\|C\|\le 1$. Introduce one mediator qubit $u$, choose
$\Delta\gg J$ and define  the simulator Hamiltonian $H=H_0+V$ with
$H_0=\Delta\,  |1\ra\la 1|_u$, \be \label{V:3to2}
V=V_d+V_{od}+V_{\rm extra}, \quad V_d=-\Delta^{\frac23}
J^{\frac13} |1\ra\la 1|_u \otimes C, \quad
V_{od}=\frac{\Delta^{\frac23} J^{\frac13}}{\sqrt{2}}\, X_u\otimes
(-A+B), \ee and  $V_{\rm
extra}=\Delta^{1/3}J^{2/3}(-A+B)^2/2+J(A^2+B^2)C/2$. Note that $H$
contains only $2$-body interactions. The  purpose of the term $V$
is to induce  transitions $|0\ra_u\to |1\ra_u \to |1\ra_u \to
|0\ra_u$ in the third order of perturbation theory, see Fig.~1,
such that the corresponding contribution to the effective
Hamiltonian is proportional to $(-A+B)^2C$ which coincides with
$ABC$ up to some unwanted terms which are canceled by $V_{\rm
extra}$. We define
 \be
S=\frac{-ix }{\sqrt{2}}  \, Y_u \otimes
(-A+B)\left[I+x C+x^2 C^2-
\frac{2 x^2}{3} (-A+B)^2\right], \quad x\equiv \left( \frac{J}{\Delta}\right)^{\frac13}.
\label{eq:s_3to2}
 \ee
One can check that $S$ is the Schrieffer-Wolff transformation
truncated at the third order, i.e., $S=S_1+S_2+S_3$, see
Eq.~(\ref{S1S2S3}) in the Appendix. We calculate the effective
Hamiltonian $Pe^S H e^{-S} P$. Let us first estimate an error
resulting from cutting off the expansion, see
Lemma~\ref{prop:LR1}. Recalling that $L=[S,\cdot]$ one gets \be
\label{eq:3bound} \| P e^{S}H e^{-S}P-{P}\left(H + [S,H] + \frac12
[S,[S,H]]\right)P|| \leq \frac16 \| P\, L^3(H)\, P \| + \frac1{24}
\| L^4(H) \|. \ee We note that $P\, L^3(H)\, P=P\, L^3(V_{od})\,
P$ since each application of $S$ flips the mediator qubit and a
non-zero contribution comes only from the terms with an even
number of flips.
 Using a bound $\|S\|=O(x)$, see Eq.~(\ref{eq:s_3to2}),
one can upper-bound the r.h.s. in Eq.~(\ref{eq:3bound}) as $O(J^{4/3}\Delta^{-1/3})$. A direct but lengthy
calculation shows that
\bea \label{transformedH3} e^S H e^{-S} & = & \left(H + [S,H] +
\frac12
[S,[S,H]]+\frac{1}{6}[S,[S,[S,H]]]\right)+O(J^{4/3}\Delta^{-1/3}) \nonumber \\
& = & \left[ \ba{cc} H_{\rm
target} & 0 \\ 0 & \Delta\, I + O(\Delta^{2/3} J^{1/3}) \\
\ea\right] + O(J^{4/3}\Delta^{-1/3}), \eea where $H_{\rm target}=J
ABC$, as desired.  The error $O(J^{4/3}\Delta^{-1/3})$ in $H_{\rm
target}$ can be made $O(\epsilon J)$ by choosing
$\Delta=J\epsilon^{-3}$.

\begin{figure}
\label{fig:gadgets}
\centerline{
\mbox{
 \includegraphics[height=2cm]{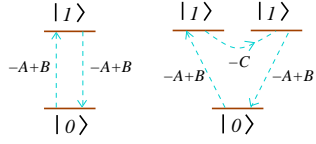}
 }}
\caption{Transitions induced by the perturbation $V$ for the
subdivision gadget (left) and the $3$-to-$2$-local gadget (right).}
\end{figure}

\section{II. Combining the gadgets together}
\label{sec:sw}

Let $H_{\rm target}=\sum_u H_{\rm target}^u + H_{\rm else}$, where
each term $H_{\rm target}^u$ can be dealt with using one of the
gadgets described above. The simulator Hamiltonian is $H=\sum_u H^u
+ H_{\rm else}$ where $H^u$ is the simulator constructed for a
gadget $u$ as above. Using
Eqs.~(\ref{transformedH2},\ref{transformedH3}) one gets the
inequality $e^{S^u} H^u e^{-S^u} \ge I^u\otimes H_{\rm target}^u +
O(\epsilon J)$ which yields $H^u\ge e^{-S^u} ( I^u\otimes H_{\rm
target}^u) e^{S^u}+ O(\epsilon J)$. Applying Lemma~\ref{prop:LR1}
one gets $\| e^{-S^u} ( I^u\otimes H_{\rm target}^u) e^{S^u} -
I^u\otimes H_{\rm target}^u\|\le \|\, [S^u,I^u\otimes H_{\rm
target}^u]\, \| =O(\epsilon J)$, see
Eqs.~(\ref{eq:s_sub},\ref{eq:s_3to2}). It follows that $H^u\ge
I^u\otimes H_{\rm target}^u + O(\epsilon J)$. Summing up these
inequalities over all gadgets one arrives at $H\ge I\otimes H_{\rm
target} + O(\epsilon n J)$, where $I$ acts on the mediator qubits.
Thus \be \label{lower_bound} \lambda(H)\ge \lambda(H_{\rm target}) +
O(\epsilon n J). \ee

\noindent {\bf  Bounding Cross-Gadget Contributions.} Let
$P=\bigotimes_{u} P^{u}$ be the projector on the subspace in which
every mediator qubit is in the state $|0\ra$. Define $S=\sum_u S^u$
where  $S^u$ is constructed using
Eqs.~(\ref{eq:s_sub},\ref{eq:s_3to2}). Below we shall prove that \be
\label{small-cross-talk} \| P\, e^{S}\, H \, e^{-S}\, P - H_{\rm
target} \| =O(\epsilon n J). \ee Then one can  get an upper bound on
$\lambda(H)$ by  observing that \be \label{upper_bound}
\lambda(H)=\lambda( e^{S}\, H \, e^{-S} ) \le \lambda(P\, e^{S}\, H
\, e^{-S}\, P)=\lambda(H_{\rm target}) +O(\epsilon n J). \ee
Combining Eqs.~(\ref{lower_bound},\ref{upper_bound}) one gets
$|\lambda(H)-\lambda(H_{\rm target})|=O(\epsilon nJ)$ which is the
desired result. It remains to prove Eq.~(\ref{small-cross-talk}).
The idea is to bound the cross-gadget terms  using Lemma
\ref{prop:LR1} and Lemma \ref{prop:LR2}. Two important properties
(valid for both gadgets) will be used repeatedly in this argument:
(1) $S^u$  always flips a mediator qubit $u$; (2)   For $u \neq v$
one has $[P^u,H^v]=0$, $[P^u,S^v]=0$, $[H_0^u,S^v]=0$. This latter
property essentially captures the independent action of the local
gadgets.

Let us start with the subdivision gadget. Recall that we have chosen
$\Delta=J \epsilon^{-2}$. From Eq.~(\ref{eq:s_sub}) we can see that
$S$ is a $O(1)$-local operator with an interaction strength
$O(\epsilon)$. Applying Lemma~\ref{prop:LR1}  we get \be
\label{sub-cross-talk}
 \| Pe^S He^{-S} P- H_{\rm target}\| \le
  \| \, [S,H_{\rm else}] \, \| + \sum_u \| \, P\,( H^u + [S,H^u]+\frac12 [S,[S,H^u]] )\, P -H_{\rm target}^u \, \| +
 \frac1{6} ||L^3(H) \|,
\ee Lemma \ref{prop:LR2} gives $||[S,H_{\rm else}]||=O(\epsilon n
J)$ and $||L^3(H)||=O(\epsilon^3\Delta n)=O(\epsilon nJ)$.
Properties~(1),(2) above imply $P\, [S^v,H^u]\, P=0$ for $u\ne v$
and that $P\, [S^{u_1},[S^{u_2},H^v]]\, P=0$ unless  $u_1=u_2$.
Property~(2) implies that $P\, [S^{v},[S^{v},H^u_0]]\, P=0$ for
$u\ne v$ and thus
 the only non-zero cross-talk term could be
$P\, [S^{v},[S^v,V^u]]\, P$ where $u\ne v$. Using the explicit form
of $V^u$, see~Eq.~(\ref{H}), and property~(1) one concludes that
$P\, [S^{v},[S^v,V^u]]\, P=P\, [S^v,[S^v,V_{\rm extra}^u]]\, P$. By
construction of the individual gadgets, the contribution of the
``diagonal" terms (those in which $H^u$ and $S^u$ belong to the same
gadget) is $O(\epsilon nJ)$, see Eq.~(\ref{transformedH2}).
Combining these observations together we arrive at \be
\label{sub-cross-talk1} \| Pe^S He^{-S} P- H_{\rm target}\| \le
O(\epsilon nJ) + \sum_{u\ne v} \| \, P\, [S^v,[S^v,V_{\rm
extra}^u]]\, P \, \| \ee Now, note that for any $u$ there exist only
$O(1)$ mediator qubits $v$ such that $[S^v,V_{\rm extra}^u]\ne 0$.
Since $V_{\rm extra}^u$ has norm $O(J)$ the contribution of the
cross-talk terms is $O(n\epsilon^2 J)$ and we arrive at
Eq.~(\ref{small-cross-talk}).

Let us do a similar analysis for the 3-to-2-local gadget.  Recall
 that we have chosen $\Delta=J \epsilon^{-3}$. From Eq.~(\ref{eq:s_3to2}) we see that $S$ is a $O(1)$-local
operator with interaction strength $O(\epsilon)$. Applying
Lemma~\ref{prop:LR1} we get \bea \| \, P\, e^S He^{-S}\, P- H_{\rm
target}\|  &\leq &
 \|\,  [S,H_{\rm else}\, \| +
\sum_u \| \, P\,( H^u + [S,H^u]+\frac12 [S,[S,H^u]] )\, P -H_{\rm target}^u \, \| \nn \\
&& +  \frac1{3!}\, \| P L^3(H)P \| +\frac1{4!} \, \| L^4(H)\|.
\label{eq:3to2b} \eea Lemma \ref{prop:LR2} gives $ \|\,  [S,H_{\rm
else}\, \| =O(\epsilon nJ)$, $ \| L^4(H)\| =O(\epsilon^4 \Delta
n)=O(\epsilon nJ)$. Property~(1) implies $P L^3(H_0)P=0$. Taking
into account that $V$ has interaction strength $\Delta^{2/3}
J^{1/3}=\Delta \epsilon$, see Eq.~(\ref{V:3to2}), and applying
Lemma~\ref{prop:LR2} we get $ \| P L^3(H)P \|= \| P L^3(V)P
\|=O(n\epsilon^3 \Delta \epsilon) = O(\epsilon nJ)$. Repeating the
same arguments as for the subdivision gadget we conclude that the
only cross-talk terms contributing to Eq.~(\ref{eq:3to2b}) are
$P\, [S^v,[S^v,V^u]]\, P$ where $u\ne v$. Using the explicit form
of $V^u$, see Eq.~(\ref{V:3to2}), one concludes that $P\,
[S^v,[S^v,V^u]]\, P=P\, [S^v,[S^v,V^u_{\rm extra}]]\, P$ so we
again arrive at the bound Eq.~(\ref{sub-cross-talk1}). By
definition, $V_{\rm extra}^u$ has norm $O(\Delta^{1/3}
J^{2/3})=O(J\epsilon^{-1})$ and $S^v$ has norm $O(\epsilon)$.
Making use of the properties of $H_{\rm target}$ we conclude that
the contribution of the cross-talk terms is $O(n\epsilon J)$ and
we arrive at Eq.~(\ref{small-cross-talk}).


\section{III. Simulation overhead}
Let us estimate the overall increase of the interaction strength
associated with $m$ levels of simulation using the subdivision
gadget. Let $H_i$ be a Hamiltonian at a level $i$ of the simulation
such that $H_0=H_{\rm target}$ and $H_m$ is $3$-local. Accordingly,
$H_{\rm target}$ can be $k$-local where $k\sim 2^m$. The overall
number of mediator qubits one needs for the simulation is
$O(n2^{O(k)})$. For simplicity one can add all these qubits to the
system from the beginning by letting $H_{\rm target}$ to act
trivially on them. It allows us to assume that at every level of
simulation we have $n'=O(n2^{O(k)})$ qubits. Let $\lambda_i$ be the
ground state energy of $H_i$. Let $J_i$ be the interaction strength
of $H_i$. In order to make $|\lambda_{i+1}-\lambda_i|\sim n'\delta$
it suffices to choose $J_{i+1}=\delta^{-2} J_i^3$. Accordingly,
$\delta^{-1} J_m =(\delta^{-1} J_0)^{3^m}$ and thus $J_{m=O(\log
k)}\sim J_0^{poly(k)} \delta^{-poly(k)}$. Choosing
$\delta=2^{-O(k)}\epsilon$ one can make the overall error
$|\lambda_0-\lambda_m|\sim \epsilon n$. Assuming that $J_0=O(1)$ we
get $J_{\rm final}\sim (2\epsilon^{-1})^{poly(k)}$.

It may be possible to improve this exponential scaling by using a
direct $k$-to-2-local gadget.


\section{Acknowledgements}
SB, DPD and BMT acknowledge support by DTO through ARO contract
number W911NF-04-C-0098. DL acknowledges support from the Swiss NSF.

\appendix

\section{Appendix}

\begin{proof}[Proof of Lemma~\ref{prop:LR1}]
For any real $t\ge 0$, let $H(t)=e^{St}\, H \, e^{-St}$. Define an
auxiliary quantity \be \label{Delta1} r_k(H,t)=\|\,
H(t)-\sum_{p=0}^{k-1} \frac{t^p}{p!} \, L^p (H)\, \|. \ee Let us
get an upper bound on the increment $r_k(H,t+\delta t)-r_k(H,t)$.
First of all notice that \be \label{aux_bound1} H(t+\delta
t)-H(t)=\delta t\, e^{St}\, L(H)\, e^{-St} + O((\delta t)^2). \ee
Secondly, for any $k\ge 1$ one has \be \label{aux_bound2}
\sum_{p=0}^{k-1} \frac{(t+\delta t)^p}{p!} \, L^p (H) -
\sum_{p=0}^{k-1} \frac{t^p}{p!} \, L^p (H)= \delta t
\sum_{p=0}^{k-2} \frac{t^p}{p!} L^p(L(H))+O((\delta t)^2). \ee (if
$k=1$ then there are no $O(\delta t)$ terms on the right hand
side). Applying the triangle inequality one gets \be
\label{increment} r_k(H,t+\delta t)-r_k(H,t)\le \delta t\,
r_{k-1}(L(H),t) + O((\delta t)^2). \ee Since $r_k(H,t)$ is a
continuous function of $t$, it is legitimate to add up the
inequalities Eq.~(\ref{increment})  for $t=0,\delta t,2\delta
t,\ldots,s$, take a limit  $\delta t\to 0$, and replace the
resulting sum by an integral, which yields \be \label{recursive}
r_k(H,s) \le \int_0^s dt\, r_{k-1}(L(H),t) \quad \mbox{if $k\ge
1$}, \quad \mbox{and} \quad r_0(H,s)=\|H\|. \ee Applying this
upper bound recursively and evaluating the integrals one arrives
at \be \label{final} r_{k}(H,t)\le \|L^{k}(H)\| \cdot
\frac{t^{k}}{(k)!} \quad \mbox{for all $k\ge 0$}. \ee Since
$r_k(H)=r_k(H,1)$ Eq. (\ref{final}) proves the lemma.
\end{proof}

\begin{proof}[Proof of Lemma~\ref{prop:LR2}]
Represent $S$ and $H$ as a sum of local operators
$S=\sum_{i=1}^{K'} S_i$ and $H=\sum_{j=1}^K H_j$ such that any
$S_i$, $H_j$ act on $O(1)$ qubits and any qubit is acted on by
$O(1)$ operators $S_i$, $H_j$. Let us use the term {\it elementary
commutator of order $k$} for a multiple commutator that involves
some $H_j$ and $k$ operators $S_i$. For example, $[S_1,[S_2,H_1]]$
is an elementary commutator of order $2$. Note that a commutator
of any $O(1)$-local operators is again a $O(1)$-local operator.
Therefore for any constant $k$ the number of non-zero elementary
commutators of order $k$ contributing to $L^k(H)$ can be bounded
as $O(n)$. Each elementary commutator of order $k$ has a norm at
most $2^k J_S^k \, J_H$. Applying the triangle inequality to $\|
L^k (H)\|$ we get Eq.~(\ref{LR2}).
\end{proof}

\section{Systematic Solution for the Schrieffer-Wolff transformation}
One can note that the equation $P\, e^S (H_0+\beta\, V)e^{-S}\, Q=0$
has a unique solution in terms of a formal Taylor series
$S=\sum_{p=1}^\infty S_p\ \beta^p$; the series coefficients $S_p$
can be straightforwardly computed. In the gadgets we have used an
approximate version of $S$ obtained by truncating the Taylor series
at the first (subdivision gadget) or the third order (3-to-2-local
gadget). Introducing operators $V_d=PVP+QVQ$, $V_{od}=PVQ+QVP$ and
linear maps $L_0=[H_0,\cdot]$, $L_p=[S_p,\cdot]$ one can derive that
\be \label{S1S2S3} S_1=L_0^{-1}(V_{od}), \quad S_2=L_0^{-1}L_1(V_d),
\quad S_3=-\frac13 \, L_0^{-1} L_1^3(H_0) + L_0^{-1} L_2 (V_d). \ee



\end{document}